\documentclass[aps,preprint,nofootinbib,preprintnumbers,eqsecnum,superscriptaddress]{revtex4-1}

\usepackage{color}

\newcommand{\nb}[1]{\color{blue}}

\newcommand{\hl}[1]{\color{magenta}}

\usepackage[
	  pagebackref=false,
	  colorlinks=true,
      linkcolor=blue,
      urlcolor=blue,
      filecolor=black,
      citecolor=red,
      pdfstartview=FitV,
      pdftitle={},
        pdfauthor={Paolo Glorioso, Michael Crossley, Hong Liu},
        pdfsubject={},
        pdfkeywords={},
        pdfpagemode=None,
        bookmarksopen=true
      ]{hyperref}

\usepackage{bbold}

\usepackage[normalem]{ulem}
\usepackage{amsmath}
\usepackage{enumerate}
\usepackage{amsfonts}
\usepackage{epsfig}
\usepackage{mathbbol}

\def\Tr{\mathop{\rm Tr}}

\newcommand\half{{\ensuremath{\frac{1}{2}}}}
\newcommand\p{\ensuremath{\partial}}

\newcommand\vev[1]{{\ensuremath{\left\langle{#1}\right\rangle}}}

\newcommand{\be}{\begin{equation}}
\newcommand{\ee}{\end{equation}}
\newcommand{\bea}{\begin{eqnarray}}
\newcommand{\eea}{\end{eqnarray}}
\newcommand{\bega}{\begin{gather}}
\newcommand{\eega}{\end{gather}}

\newcommand{\bi}{\begin{itemize}}
\newcommand{\ei}{\end{itemize}}
\newcommand{\ben}{\begin{enumerate}}
\newcommand{\een}{\end{enumerate}}
\newcommand{\bca}{\begin{cases}}
\newcommand{\eca}{\end{cases}}
\newcommand{\bln}{\begin{align}}
\newcommand{\eln}{\end{align}}
\newcommand{\bst}{\begin{split}}
\newcommand{\est}{\end{split}}
\def\ie{\begin{equation}\begin{aligned}}
\def\fe{\end{aligned}\end{equation}}
\newcommand{\bma}{\le(\begin{matrix}}
\newcommand{\ema}{\end{matrix}\ri)}

\newcommand\al{{\alpha}}
\def\b{{\beta}}
\newcommand\ep{\epsilon}
\newcommand\sig{\sigma}

\newcommand\Lam{\Lambda}
\newcommand\om{\omega}

\newcommand\ga{{\ensuremath{{\gamma}}}}
\newcommand\Ga{{\ensuremath{{\Gamma}}}}
\newcommand\de{{\ensuremath{{\delta}}}}

\newcommand\vp{\varphi}
\newcommand\ka{\kappa}
\newcommand\ze{\zeta}
\newcommand\vep{\varepsilon}

\newcommand\da{{\dagger}}

\newcommand\Th{{\Theta}}
\def\th{{\theta}}

\newcommand\ov{\over}
\newcommand\ha{{\half}}

\def\le{\left}
\def\ri{\right}

\newcommand\sC{{\ensuremath{{\mathcal C}}}}

\newcommand\sF{{\ensuremath{{\mathcal F}}}}

\newcommand\sG{{\ensuremath{{\mathcal G}}}}

\newcommand\sL{{\ensuremath{{\mathcal L}}}}

\newcommand\sO{{\ensuremath{{\mathcal O}}}}
\newcommand\sP{{\ensuremath{{\mathcal P}}}}

\newcommand\sT{{\mathcal T}}

\newcommand\vx{{\vec x}}
\newcommand\vk{{\vec k}}

\newcommand{\hmu}{{\hat \mu}}

\newcommand{\rmi}{{\rm i}}
\newcommand{\rmj}{{\rm j}}

 \global\long\def\mL{\mathcal{L}}

 \global\long\def\e{\epsilon}

\global\long\def\f#1#2{\frac{#1}{#2}}

 \global\long\def\b{\beta}
 
 \global\long\def\G{\Gamma}

 \global\long\def\d{\delta}
 \global\long\def\Tr{\text{Tr}}

 \global\long\def\T{\text{T}}
 \global\long\def\w{\omega}

 \global\long\def\sC{\mathscr{C}}

 \global\long\def\sF{\mathscr{F}}
 \global\long\def\sG{\mathscr{G}}

 \global\long\def\sO{\mathscr{O}}

 \allowdisplaybreaks
 \usepackage{verbatim}
 \usepackage{amssymb,amsmath,amsthm}

\newtheorem*{theorem}{Theorem}

\usepackage{bbm}
\usepackage{mathbbol}
\usepackage[mathscr]{eucal}

\begin{document}

\title{Ghostbusters: Unitarity and Causality of Non-equilibrium Effective Field Theories}

\preprint{MIT-CTP/4994}
\preprint{EFI-18-4}

\author{Ping Gao}
\affiliation{Center for the Fundamental Laws of Nature,
Harvard University,
Cambridge, MA 02138}

\author{Paolo Glorioso}
\affiliation{Kadanoff Center for Theoretical Physics and Enrico Fermi Institute\\
University of Chicago, Chicago, IL 60637, USA}

\author{Hong Liu}
\affiliation{Center for Theoretical Physics, \\
Massachusetts
Institute of Technology,
Cambridge, MA 02139 }

\begin{abstract}

\noindent  For a non-equilibrium physical system defined along a closed time path (CTP), a key constraint is the so-called largest time equation, which is a consequence of unitarity and implies causality. In this paper, we present a simple proof that if the propagators of a non-equilibrium effective action have the proper pole structure, the largest time equation is obeyed to all loop orders. Ghost fields and BRST symmetry are not needed. In particular, the arguments for the proof can also be used to show that if ghost fields are introduced, their contributions vanish.

\end{abstract}

\today

\maketitle

\tableofcontents

\section{Introduction}

Effective field theories (EFT) provide powerful tools for dealing with many problems in condensed matter and
particle physics.
Recently we have applied the EFT approach to  local equilibrium processes 
to find a new proof of the second law of thermodynamics~\cite{GL}, and a new formulation of fluctuating hydrodynamics~\cite{CGL,CGL1,Glorioso:2017lcn}. Essential elements of the formulation include  various constraints from unitarity and a $Z_2$ dynamical KMS symmetry which imposes micro-time-reversibility and local equilibrium. These EFTs provide a first-principle derivation of and   systematize the phenomenological Martin-Siggia-Rose-De Dominicis-Janssen~\cite{msr,Dedo,janssen1} functional integral approaches to various  stochastic
equations.


In these work it was also realized one of the unitarity constraints could potentially be violated from loop corrections.
Anticommuting ghost variables and BRST symmetry were then introduced to make sure the unitarity constraint is maintained~\cite{CGL} (see also~\cite  {Haehl:2015foa,Haehl:2015uoc,Haehl:2016pec,yarom,Jensen:2018hhx}).
Intriguingly, it can be further shown that when the BRST symmetry is combined with the $Z_2$ dynamical KMS symmetry, there is always an emergent supersymmetry~\cite{CGL,GaoL}\footnote{This extends many previous work on emergent supersymmetry in stochastic systems~\cite  {parisi,feigelman,Gozzi:1983rk,Mallick:2010su,zinnjustin}. See also~\cite{Haehl:2015foa,Haehl:2015uoc,Haehl:2016pec,yarom,Jensen:2018hhx} which used supersymmetry as an input for constructing an action principle for hydrodynamics.}.

In this paper we further clarify the fate of the unitarity constraint  under loop corrections with a new piece of information
which was not considered in~\cite{CGL}: a  set of propagators of such an EFT must be retarded. Given the retarded nature of these propagators, we will then be able to prove, with an appropriate choice of the regularization procedure, to all loop orders that: (i) Even in the absence of ghost variables, unitarity and causality are maintained; 
(ii) Integrating out the ghost action results in no contributions. 
Thus ghost variables are not needed.

The fact that the retarded structure of the propagator causes ghost diagrams to vanish is well-understood in the context of the Langevin equation, see e.g.~\cite{Arnold,Gonzalez}. The present work can be seen as the extension of such results to non-equilibrium EFTs defined on a closed time path (or in the Schwinger-Keldysh formalism).  The central relation about unitarity and causality we prove is the so-called largest time equation (LTE). This was originally formulated by Veltman and 't-Hooft for quantum field theory at zero temperature \cite{veltman,thooft}, and was generalized to finite temperature by subsequent works \cite{kobes,aurenche,Gelis:1997zv,Bedaque:1996af,Caron-Huot:2007zhp}. Our proof establishes that LTE holds for general non-equilibrium EFTs. 
A step in this direction was also taken recently in~\cite{loga}. 

In the formulation of~\cite{GL,CGL,CGL1}, the retarded nature of the propagators 
is a consequence of the $Z_2$ dynamical KMS symmetry and unitarity constraints,
and reflects the coincidence of thermodynamic and causal arrow time~\cite{GL}.
More explicitly,  it means that dissipative coefficients of the action must have the ``right'' signs--for example, friction coefficients, viscosities, conductivities must be non-negative--which ensures that on the one hand entropy increases monotonically with time,
and on the other hand the system is causal.







The plan of the paper is  as follows. In next section we review key elements of local equilibrium EFTs. In Sec.~\ref{sec:pert} we discuss the structure of the perturbative action. In Sec.~\ref{sec:LTE} we prove that the theory satisfies the LTE. In Sec.~\ref{sec:ghost} we show that the ghost contribution is zero.
In Sec.~\ref{sec:uni} we show that other unitarity constraints are also satisfied. We conclude in Sec.~\ref{sec:conc} with some general remarks.

\section{EFTs for local-equilibrium systems}  \label{sec:NEQFT}

In this section we review some essential aspects of the local equilibrium EFT formulated in~\cite{CGL,CGL1,GL}.

\subsection{Constraints on a CTP generating functional}

Consider the generating functional for a quantum statistical system defined on
a closed time path  (CTP) contour~\cite{schwinger,keldysh,Feynman:1963fq},
\bln \label{gen0}
e^{W [\phi_1, \phi_2]}  & = \Tr \le( U (+\infty, -\infty; \{\phi_{1i}\}) \rho_0 U^\da (+\infty, -\infty; \{\phi_{2i}\}) \ri) \cr
& =  \Tr \le[\rho_0 \sP \exp \le(i \int d t \, (\sO_{1i} (t) \phi_{1i} (t) - \sO_{2i} (t) \phi_{2i} (t)) \ri)\ri] \cr
& = 
 \Tr \le[\rho_0 \sP \exp \le(i \int d t \, (\sO_{ri} (t) \phi_{ai} (t)  +  \sO_{ai} (t) \phi_{ri} (t)) \ri)\ri]
\end{align}
where $\rho_0$ denotes the initial state of the system, and $U (t_2, t_1; \{\phi_i \})$ is the evolution operator of the system
from $t_1$ to $t_2$ in the presence of external sources  $\{\phi_i\}$.\footnote{The sources are assumed to have compact support in spacetime. 
}  In the above equations we have suppressed spatial dependence (and spatial integrals) for notational simplicity and will also do so below.  Note that $\sO_{1i}$ and $\sO_{2i}$
are the same operator with subscripts $1,2$  indicating only the segments of the contour they are inserted.
In the last line we also introduced the so-called $r-a$ variables
\be \label{ra1}
\phi_{ri} = \ha (\phi_{1i} + \phi_{2i}), \quad \phi_{ai} = \phi_{1i} - \phi_{2i}, \quad
\sO_{ai} = \sO_{1i} - \sO_{2i}, \quad \sO_{ri} = \ha ( \sO_{1i} + \sO_{2i})\ .
\ee
$n$-point functions of  $\sO_{ri}, \sO_{ai}$ are obtained by taking functional derivatives of the generating functional $W$,
\be \label{dfun3}
G_{\al_1 \cdots \al_n  } (t_1, \dots, t_n) \equiv {1 \ov i^{n_r}}
{\de^n W \ov \de \phi_{\bar \al_1 } (t_1) \cdots
\de \phi_{\bar \al_n  } (t_n)} \biggr|_{\phi_{a} = \phi_{r} =0} = i^{n_a} \vev{\sP \sO_{\al_1} (t_1) \cdots \sO_{\al_n} (t_n)}
\ee
where $\al_1, \cdots , \al_n \in (a, r)$ and $\bar \al = r, a$ for $\al = a, r$. $n_{r,a}$ is the number of $r$ and $a$-index in $\{\al_1, \cdots, \al_n\}$ respectively. In~\eqref{dfun3} for notational simplicity, we have suppressed indices $i$ labelling different operators
and will also do so below.

Due to the unitary nature of evolutionary operator $U$, the generating functional $W$ satisfies a number of constraints (taking $\phi_{1,2}$ to be real)
\bega \label{Wc1}
  W^* [\phi_{r}, \phi_{a}] =  W [\phi_{r}, - \phi_{a}]   \\
 \label{Wc2}
{\rm Re} \, W \leq 0 \\
 \label{Wc3}
W [\phi_{r}, \phi_{a} =0]  = 0 \
\end{gather}
which can be readily seen from~\eqref{gen0}. Equation~\eqref{Wc3} means that correlation functions involving only $\sO_a$ are identically zero, i.e.
\be \label{Wc4}
G_{a \cdots a} (t_1, \cdots t_n) = 0 \ .
\ee

In fact,~\eqref{Wc4} can be further strengthened to obtain
the LTE~\cite{veltman,thooft,kobes,aurenche,Gelis:1997zv,Bedaque:1996af,Caron-Huot:2007zhp}
\be \label{lte1}
G_{\alpha_1\cdots \alpha_n  a}(t_1,\cdots,t_n,t_{n+1})=0,\qquad \text{if}\quad t_{n+1}>t_1,\cdots,t_n \
\ee
where $\al_i, i=1, \cdots n$ can be either $r$ or $a$. LTE says that $G$ is identically zero whenever the operator with the largest time is an $a$-type operator.
Clearly~\eqref{Wc4} is a subcase of~\eqref{lte1}. To see~\eqref{lte1}, suppose $\phi_r(t)=\phi (t)$ and $\phi_a (t) =0$ for $t>t_{\rm max}$
(where $t_{\rm max} = {\rm max} (t_1,\cdots,t_n)+\e$ for small and positive $\e$). In~\eqref{gen0}, due to unitarity of $U$, the parts after $t=t_{\rm max}$ of the evolution operators will cancel between the upper and the lower branches of the contour, so that the generating functional $W$ will be independent of  the values of $\phi_r = \phi (t)$ for $t > t_{\rm max}$, which immediately implies (\ref{lte1}). {See Fig.~\ref{fig:lte}a.} Alternatively, equation~\eqref{lte1} is equivalent to the statement that for the operator with the largest time, it does not matter whether one inserts it on the upper or lower contours. {See Fig.~\ref{fig:lte}b.}

\begin{figure}
\begin{center}
\includegraphics[width=13cm]{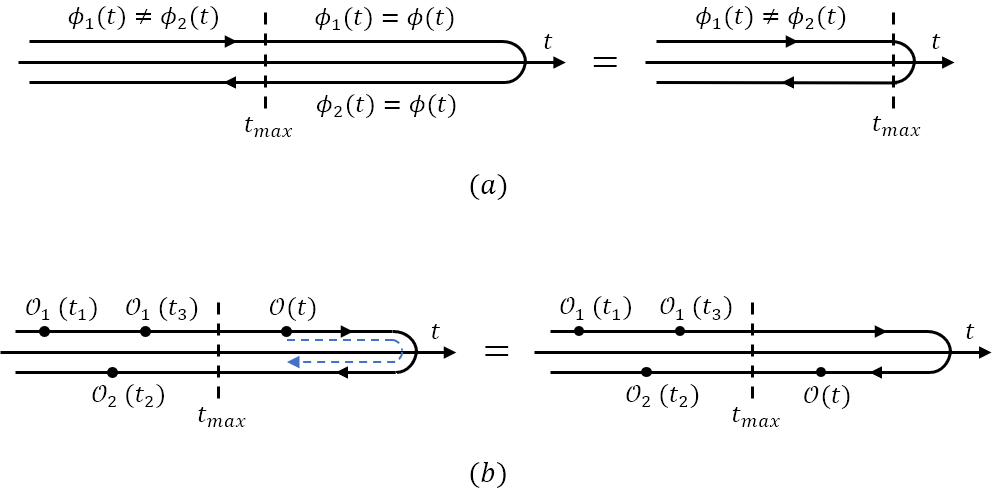}
\caption{Largest time equation. (a) From the perspective of external sources: the path integral is independent of the part where the external sources for upper and lower contours are the same;  (b) From the perspective operator insertion: for the operator with the largest time it does not matter whether one inserts it on the upper or lower contour.}\label{fig:lte}
\end{center}
\end{figure}

Equation~(\ref{lte1}) can also be seen as a statement of causality. For example, considering $n=1$, in addition to $G_{aa} =0$, we have
\be \label{onk}
G_{ra} (t_1, t_2) = 0, \quad t_2 > t_1 \ .
\ee
Recall that $G_{ra} (t_1, t_2) = G_R (t_1 - t_2)$ where $G_R$ is the retarded Green's function. So~\eqref{onk} simply says
that $G_R (t) = 0$ for $t < 0$, which is the statement that responses should come after disturbances.
Similarly for general $n$, $G_{\al_1 \cdots \al_n a}$ can be considered as the response function for $G_{\al_1 \cdots \al_n}$ with $t_{n+1}$ as the time for turning on disturbances. Thus~\eqref{lte1} again says responses cannot come before disturbances.

\subsection{General structure of local equilibrium EFTs}  \label{sec:eft}

We will now focus on local equilibrium systems for which macroscopic physical quantities and the external sources vary over spatial and time scales much larger than microscopic relaxation scales\footnote{For language and notational simplicity we will use $\ell$ to denote both relaxation time and length, which can of course be in principle independent.} $\ell$.
We can then imagine integrating out all degrees of freedom whose characteristic spacetime scales are smaller than $\ell$
and express~\eqref{gen0} as
\be
e^{W [\phi_1, \phi_2]}  =  \int_{\ell} D \chi_1 D \chi_2 \, e^{i I_{\rm eff} [\chi_1, \phi_1; \chi_2 , \phi_2; \rho_0]}   \
\label{left}
\ee
where $\chi_{1,2}$ denote the remaining ``slow'' variables  (which we will take to be real) and there are again two copies of them.
It is convenient to introduce
\be \label{rava}
\chi_r = \ha (\chi_1 + \chi_2) , \quad \chi_a = \chi_1 - \chi_2, 
\ee
where $\chi_r$ are usually interpreted as physical variables  
while $\chi_a$ as noises. 

Similarly to~\eqref{Wc1}--\eqref{Wc4}, unitarity of time evolution in~\eqref{gen0} imposes nontrivial constraints on $I_{\rm eff}$ (see e.g. Appendix A of~\cite{GL} for a derivation)
\bega \label{fer1}
  I^*_{\rm eff} [\chi_r , \phi_r; \chi_a, \phi_a] = - I_{\rm eff} [\chi_r , \phi_r; -\chi_a, -\phi_a]   \\
 \label{pos}
{\rm Im} \, I_{\rm eff} \geq 0 \\
 \label{key1}
I_{\rm eff} [\chi_r =\chi, \phi_r = \phi; \chi_a=0, \phi_a =0] = 0 \ .
\end{gather}
$ I_{\rm eff}$ is generically complex and equation~\eqref{fer1} implies that terms in $ I_{\rm eff}$ which are {\it even} in $a$-variables 
must be pure imaginary. 
 Equation~\eqref{key1} implies that any term in the action must contain at least one factor of $a$-type variables ($\phi_a$ or $\chi_a$).

 Furthermore, we require $I_{\rm eff}$ satisfy an anti-linear $Z_2$ dynamical KMS symmetry (in the absence of background sources) 
\be \label{lkms}
I_{\rm eff} [\chi_r, \chi_a] =I_{\rm eff}   [\tilde \chi_r, \tilde \chi_a]
\ee
where 
$\tilde \chi_r, \tilde \chi_a$ denote the transformed variables which in the classical limit can be written schematically as\footnote{
The explicit expressions for various theories are given in Sec.~\ref{sec:exm}.}
\be \label{dkms}
\Th  \tilde \chi_r = \chi_r + \cdots, \qquad \Th  \tilde \chi_a = \chi_a + i \Phi_r  + \cdots
\ee
where $\Th$ is a discrete spacetime reflection involving time reversal (it can be $\sT$ or $\sP \sT$ or $\sC \sP \sT$ and so on), $\Phi_r$ denote some expression of $r$-variables with a {\it single} derivative, and $\cdots$ denote $O(\hbar)$ corrections. The dynamical KMS symmetry imposes
micro-time-reversibility and local equilibrium~\cite{GL,CGL,CGL1}.  

An immediate question is whether the path integral~\eqref{left} satisfies~\eqref{Wc1}--\eqref{Wc3} and~\eqref{lte1}
given~\eqref{fer1}--\eqref{key1}.  Examination of~\eqref{left} perturbatively appeared to indicate  
that~\eqref{Wc3} could be potentially violated by loop corrections.
In~\cite{CGL} it was proposed to associate with each pair of $\chi_{r,a}$ a pair of anti-commuting ghost variables $c_{r,a}$, and require that when $\phi_a=0$, the full action be invariant under the following BRST-type symmetry
\be \label{omm}
\d\chi_{r} \equiv \e Q \chi_{r}
=\e c_{r}, \quad\d c_{a}
\equiv \e Q\chi_{a}
=\e \chi_{a}, \quad \d\chi_{a}=\d c_{r}=0 \ .
\ee
Such an action can be written as $I_{\rm eff} = QV[\phi_r,\chi_{r,a},c_{r,a}]$ for some functional $V$. The variation of generating functional with respect to source $\phi_r$ is
\be \label{W-inv-brst}
\d e^{W[\phi_a=0]}=i \int D\chi_{r,a}Dc_{r,a}\, e^{iI_{\text{eff}}} \, Q \d V= i \int D\chi_{r,a}Dc_{r,a} \, Q \left( e^{iI_{\text{eff}}} \d V\right)=0,
\ee
where in the last step, we used the fact that $Q$ is a derivative on $\chi_{r}$ and $c_{a}$. Since $W[\phi_a=0]$ is independent on $\phi_r$, we can normalize it to be zero.


 By extending the dynamical KMS transformation~\eqref{dkms} of bosonic fields to ghosts for a BRST invariant action, one
 finds that under $Z_2$, charge $Q$ is mapped to a new conserved ``mirror" charge  $\bar{Q}$, and $Q, \bar Q$ together
form a supersymmetric algebra. Furthermore,~\citep{GaoL} shows that such supersymmetric extension exists for any dynamical KMS invariant bosonic effective action. Note that the converse statement is, however, not true;
supersymmetry itself cannot guarantee the whole dynamical KMS symmetry.

In Sec.~\ref{sec:LTE} we will consider~\eqref{lte1} (and thus including~\eqref{Wc3}) with a new piece of information which was not considered earlier in~\cite{CGL}:  the propagators from $I_{\rm eff}$ between $\chi_r$ and $\chi_a$ are retarded. We will see this will ensure~\eqref{lte1}.




\section{Structure of perturbative action} \label{sec:pert}

To prepare for the proof of LTE of Sec.~\ref{sec:LTE}, in this section we discuss the general structure of the perturbative action of the EFT discussed in last section.
We first present a general discussion and then give some explicit examples.

\subsection{General discussion}

Consider expanding a local equilibrium EFT around some physical background, which solves equations of motion of $I_{\rm eff}$.
For simplicity we will set the sources $\phi_i$ to zero. The precise nature of the background solution is not important, which can be either an equilibrium or
non-equilibrium configuration. Equation~\eqref{key1} and the boundary conditions for the path integral~\eqref{left}
imply that a solution of equations of motion always has $\chi_a = 0$. This means that the perturbative action expanded around such a solution again satisfies~\eqref{key1}.

The perturbative Lagrangian density can be written schematically as
\be \label{lamn}
\sL = \sL_0 + \sL_I, \qquad \sL_0 =  \chi_{ai} K_{ij} \chi_{rj} + {i \ov 2} \chi_{ai} G_{ij} \chi_{aj}
\ee
where $\chi_{ai}, \chi_{ri}$ now denote deviations around the background solution. $\sL_0$ denotes the ``free''  Lagrangian with $K_{ij}$ and $G_{ij}$ some differential operators, 
and $\sL_I$ denotes ``interaction terms.'' We will specify the separation between ``free'' and ``interaction'' terms more explicitly below.
Note that $\sL_I$ must contain at least one factor of $\chi_a$.  The propagators following from $\sL_0$ are~(below $x = (t, \vec x)$)
\bega \label{p1}
G^R_{ij} (x) = \vev{\chi_{ri} (x) \chi_{aj} (0)} =i K^{-1}_{ij}  , \qquad  G^A_{ij} (x)= \vev{\chi_{ai} (x) \chi_{rj} (0)} =i \left(K^{-\T *}\right)_{ij} , \\\vev{\chi_{ai}(x) \chi_{aj}(0)} =0, \qquad  \vev{\chi_{ri}(x) \chi_{rj}(0)} = (K^{-1} G K^{-\T *})_{ij}
\label{p2}
\end{gather}
where $\T$ denotes switching indices $i$ and $j$, $G^R$ and $G^A$ are retarded and advanced Green's functions of $\chi$ respectively, and $K^*$ denotes the operator obtained from $K$ by integration by parts, namely changing $\p_\mu$ to $-\p_\mu$.

Physical observables can also be expanded perturbatively in
$\chi_{r,a}$.  For an observable $\sO$,  the corresponding $\sO_r$ contains terms without any $\chi_a$, while $\sO_a$ must contain at least one factor of $\chi_a$.  
Thus their expansions have the schematic form (where we have suppressed $i,j$ indices)
\be \label{opn}
\sO_r = L_0 \chi_r + L_1 \chi_a + \cdots , \qquad \sO_a = N_0 \chi_a  + \cdots
\ee
where $\cdots$ contain terms at least quadratic in $\chi$, $L_0, L_1, N_0$ are differential operators, and
each term in $\cdots$ of $\sO_a$ contains at least one $\chi_a$. To leading order in $\chi$, the retarded Green's function of $\sO$ is then given by
\be \label{mno}
\sG_R (x)= \vev{\sO_r (x) \sO_a (0} =  L_0 N_0^* G_R =  iL_0 N_0^* K^{-1}  \ .
\ee
Thus up to differential operator $L_0 N_0^*$, the retarded function for $\sO$ is given by that of $\chi$.

The general structure~\eqref{lamn}--\eqref{mno} of course also applies to the case that $\sL$ is the microscopic
Lagrangian of a system defined on a CTP.

For a system to be causal, $G_R (x)$ and $\sG_R (x)$ must be proportional to $\th (t)$, i.e. responses must come after
disturbances, which requires that in momentum space $K^{-1}$  only has poles in the lower half complex frequency plane. 
In other words, 
 the $r-a$ propagator must be proportional to $\th (t-t')$ while the $a-r$ propagator must be proportional to $\th (t'-t)$. See Fig.~\ref{fig:feyn}. %

\begin{figure}[h]
\begin{centering}
\includegraphics[width=16cm]{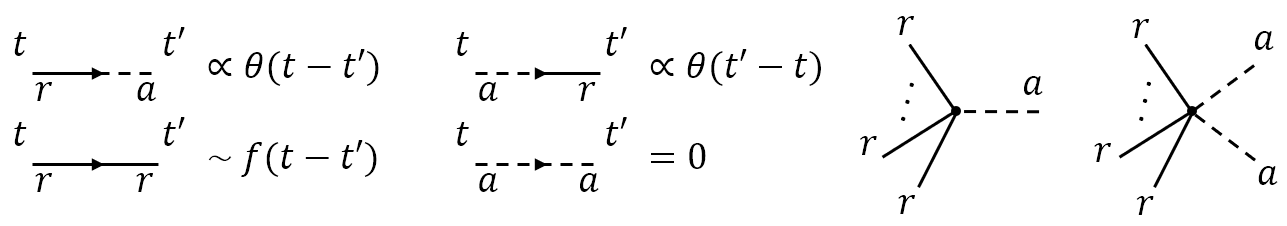}
\end{centering}
\caption{Schematic Feynman rules for~\eqref{lamn}. Dashed lines correspond to $\chi_a$'s,
solid lines correspond to $\chi_r$'s.
The $r-a$ propagator should be proportional to $\th (t-t')$ while the $a-r$ propagator should be proportional to $\th (t'-t)$. $r-r$ propagator (with some function $f$) is non-vanishing for either orderings of $t$ and $t'$. All the interacting vertices contain at least one $a$-leg. }\label{fig:feyn}
\end{figure}

Let us make some further general remarks:



\ben

\item For the case  that $\sL$ is the microscopic
Lagrangian of a system, which in general does not have any dissipative terms, the retarded nature of $K^{-1}$ should follow from appropriate $i \ep$ prescription.

\item For an effective field theory written in a derivative expansion, we include in $\sL_0$ only leading quadratic terms in the expansion. Higher derivative quadratic terms (which are suppressed by UV cutoffs compared with leading terms) as well as
nonlinear terms are in $\sL_I$.



It can be shown explicitly that the poles from leading derivative terms lie on the lower half frequency plane 
as a consequence of the $Z_2$ dynamical KMS symmetry and unitarity~\cite{GL}, as will be reviewed  in next subsection.

\item In some situations, one may be able to extend the validity of an EFT to some higher cutoff scales so that one can treat
$\sL$ non-perturbatively in derivatives. In this case one should include all quadratic terms in $\sL_0$ with
 $K, G$ interpreted as nonlocal kernels\footnote{It may also happen they only contain a finite number derivatives due to other suppressions or symmetries.}. An example is~\cite{blake}.
 In such a case,  $K^{-1}$ may have an infinite number of poles.
  While we do not have a rigorous proof, it should follow from the non-perturbative  proof of the second law of thermodynamics in~\cite{GL} and the assumption of coincidence of thermodynamic and causal arrow of time, that
the $Z_2$ dynamical KMS symmetry and unitarity are again enough to ensure the poles of $K^{-1}$ to lie in the lower half frequency plane.

An interesting caveat is the recent discussion in~\cite{blake} for quantum chaotic theories. There from a shift symmetry in the effective action, $K^{-1}$ can have a pole in upper half $\om$ plane  which gives rises to the exponential Lyapunov growth, but $\sG_R$ for physical observables such as energy density and energy fluxes have only poles in the lower half frequency plane. Even in that case as discussed in~\cite{blake},  one should deform the integration contour in Fourier transform so that $G_R (x)$ is still retarded, i.e. proportional to $\th (t)$.

\een
For the rest of the paper we will focus on the situations of items $1$ and $2$ above, and take that $K_{ij}^{-1} $ only has poles in lower half frequency plane.

Loop integrals will typically involve divergences in both frequency and spatial momentum integrations, and
regularizations are needed. For spatial momentum integrations we can use the standard dimensional regularization.
For frequency integrals in order to keep the retarded structure we will use the following
regularization 
\be \label{mm}
K_{ij}(\omega,\vec q)\to K_{ij}(\omega,\vec q)\left(\frac{\Lambda-i\omega}\Lambda\right)^p,
\ee
where $\Lam$ is the UV energy cutoff and $p$ is an appropriately chosen integer.  Clearly the regulated propagator still only has poles in the lower half frequency plane. Note that $G_{ij}$ as well as $L_0,L_1,\dots$ should be also be correspondingly modified so as to satisfy the dynamical KMS symmetry~~\eqref{lkms}--\eqref{dkms}. 
A sharp cutoff in frequency integrals will not be compatible with the dynamical KMS symmetry as~\eqref{dkms} involves time derivatives. Dimensional regularization for frequency integrals is incompatible with retarded nature of $K$.



In our discussion of subsequent sections we will need to use the $a-r$ propagator $G^A_{ij}$ in coordinate space with various possible derivatives (from interacting vertices) acting on it. More explicitly,  with~\eqref{mm} 
\be\begin{split}\label{areq1}
\p_{\mu_1}\cdots \p_{\mu_n} G_{ij}^A (t, \vec x) &= i^{n+1} \int \frac{d\omega d^{d-1}k}{(2\pi)^d}\, e^{-i\omega t+i\vec k\cdot\vec x}
K^{-1}_{ji} (-\om, - \vec k) k_{\mu_1} \cdots k_{\mu_n}
\left(\frac \Lambda{\Lambda + i\omega}\right)^p \
\end{split}
\ee
Since $K^{-1} (-\om, -\vk)$ has poles only in the upper complex $\om$-plane and so is the regulator, the above expression is proportional to $\th (-t)$ for $t > 0$ as one can close the $\om$ integration contour in the lower half complex $\om$-plane. At $t=0$, since $K_{ij}$ are polynomials in derivatives, for any fixed number $n$ derivatives, we can always choose a large enough integer $p$ so that the frequency integral in~\eqref{areq1} is well defined for $t =0$.  We can then close the path of the $\omega$-integral along the lower semicircle in the complex plane, which again gives zero. It then follows that with our choice of regularization
\be \label{rap}
\p_{\mu_1}\cdots \p_{\mu_n} G_{ij}^A (x) = \th (-t) F_{\mu_1 \cdots \mu_n ij} (x)
\ee
with
\be \label{rap1}
F_{\mu_1 \cdots \mu_n ij} (t= 0, \vx) = 0 \ .
\ee
Note that the above expression is valid for the order of limits of  taking $t \to 0$ first and then $\Lam \to \infty$ (and not the other way around). Parallel statements can be made about the $r-a$ propagator with $\th (-t)$ above replaced by $\th (t)$.



\subsection{Some explicit examples} \label{sec:exm}

Now we discuss some explicit examples of local equilibrium EFTs. Such systems can be separated into three classes: (i)  without conserved quantities; (ii) with conserved quantities; (iii) with both conserved and non-conserved quantities. To understand the pole structure for $K^{-1}$ of~\eqref{lamn} it is enough to examine a representative example in the first two classes as theories in the third class  have the same quadratic propagators as those in (i) and (ii).
For simplicity we will consider the classical limit and turn off the external sources.

\subsubsection{Brownian motion}\label{sec:br}

As the simplest example, let us consider Brownian motion of a heavy particle in a thermal medium with inverse temperature $\b_0$. The effective Lagrangian of the
particle can be written as
\be \label{Leff-br}
\sL_{\rm eff} = x_a (\p_t^2 + \nu \p_t ) x - x_a F (x)  +  {i \ka \ov 2}  x_a^2
\ee
where $\nu$ is the friction coefficient, $\ka$ controls the fluctuation of noise $x_a$, and $F(x)$ is a conservative force acting on the
particle. The dynamical KMS transformation can be written as
\be \label{km1}
\tilde x(t) = x(-t), \qquad  x_{a } (- t) = x_{a} (t) + i \beta_0  \p_t x(t) \ .
\ee
Invariance under~\eqref{km1} requires the Einstein relation
\be \label{ein}
\nu = {\b_0 \ov 2} \ka
\ee
while~\eqref{pos} requires that $\ka \geq 0$. We thus conclude that the friction coefficient $\nu$ is non-negative. We can further expand $F$ in series of $x$ as\footnote{Here we can shift $x$ by a constant to get rid of constant term in $F(x)$}
\be
F(x)=f_1 x+f_2 x^2+\cdots
\ee
where $f_i$ are constants. Above perturbative expansion holds only when $x=0$ is a stable solution. It follows that $f_1<0$ because $F$ must be a restoring force around $x=0$. The quadratic order of \eqref{Leff-br} is
\be
\sL_{0} = x_a (\p_t^2 + \nu \p_t -f_1) x - x_a F (x)  +  {i \ka \ov 2}  x_a^2
\ee
which implies that the retarded propagator in frequency space is
\be\label{brpr}
G_R(\w)=iK^{-1}(\w)=\f{-i}{\w^2+i\nu\w+f_1}
\ee
which has two poles of $\w$ at
\be
\w_{\text{pole}}=-\f{i\nu}{2}(1\pm \sqrt{1+4f_1/\nu^2})
\ee
It is clear that these two poles are both in lower half plane since $\nu$ is nonnegative.

\subsubsection{Model A}

Our next example concerns  the critical dynamics of a $n$-component real order parameter~(not conserved) $\chi_\rmi, \rmi=1,\cdots , n$ at some inverse temperature $\beta_0$ in a $d$-dimensional spacetime (i.e. model A~\cite{hohenberg,Folk}).  

The dynamical KMS transformations~\eqref{dkms} have the form (taking $\Th = \sP \sT$)
\bega \label{ckms1}
\tilde \chi_{r \rmi} (x)  = \chi_{r  \rmi}(-x) , \qquad \tilde \chi_{a \rmi} (- x) = \chi_{a \rmi} (x) +{ i \beta_0  } \p_t \chi_{r \rmi} (x) \ 
\end{gather}
and the effective Lagrangian which satisfies~\eqref{fer1}--\eqref{lkms} and is invariant under~\eqref{ckms1} can be written as\footnote{As~\eqref{ckms1} only involves time derivative we can treat time and spatial derivative expansions separately.}
\be \label{lag1}
\sL_{\rm eff}=\le(- {\de \sF \ov \de \chi_{r \rmi}}  -  \beta_0 f^{\rmi \rmj}  \p_t \chi_{r \rmj} \ri) \chi_{a \rmi}
+ {i } f^{\rmi \rmj} (\chi_r)   \chi_{a \rmi}   \chi_{a \rmj }  +  \cdots \ .
\ee
In~\eqref{lag1}, $\sF (t; \chi_r]$ is a local functional of the form
\be \label{ii1}
\sF (t; \chi_r] = \int d^{d-1} \vec x \, F (\chi_r (x), \p_i \chi_r (x), \cdots)
\ee
with $F$ an arbitrary function of $\chi_r$ and their spatial derivatives, and
$f^{\rmi \rmj} = f^{\rmj \rmi}$ are functions of $\chi_r$ satisfying, for arbitrary $a_\rmi (x) $,
\be \label{posi1}
f^{\rmi \rmj} (\chi_r)  a_\rmi (x) a_\rmj (x) \geq 0 \ .
\ee

Let us now consider a phase whose equilibrium configuration has $\chi_{r \rmi} =0$.
Keeping only quadratic terms in~\eqref{lag1}, we can write
\be
f^{\rmi \rmj} = {1 \ov \Ga_0 } \de^{\rmi \rmj}  , \qquad
F =  \ha r \chi_{r \rmi}^2 + \ha  (\p_i \chi_{r\rmi})^2 + \cdots
\ee
where we have only kept two spatial derivatives in $F$.  $\Ga_0$ should be non-negative due to~\eqref{posi1}.
$F$ (and thus the constant $r$) should also be non-negative to ensure thermodynamic stability. 
At quadratic order,~\eqref{lag1} can then be written as
\be
\sL_0 = \le( - r \chi_{r \rmi} +  \p_i^2 \chi_{r \rmi}  -  {\beta_0 \ov \Ga_0}  \p_0 \chi_{r \rmj}  \ri) \chi_{a \rmi}
+ {i \ov \Ga_0 }  \chi_{a \rmi}   \chi_{a \rmi} + \cdots  \ . 
\ee
which implies that the retarded propagator in Fourier space is
\be
G_{R,\rmi\rmj}(\w,\vec{k})=iK^{-1}_{\rmi\rmj}(\w,\vec{k})=\f{i\d_{\rmi\rmj}}{-r-k^2+i\b_0\w/\G_0}
\ee
which has pole of $\w$ at $-i\G_0(r+k^2)/\b_0$ in the lower half plane.

\subsubsection{Fluctuating hydrodynamics for relativistic charged fluids}\label{sec:fl}

Our last example is the fluctuating hydrodynamics of a relativistic charged fluid in a $d$-dimensional spacetime.
This example is a bit complicated. We will only present the final result. More details can be found~\cite{CGL1}.

Here the dynamical variables are hydrodynamical modes associated with the stress tensor and a conserved
$U(1)$ current. The $r$-variables can be chosen to be local velocity $u^\mu (x)$, local inverse temperature $\beta (x)$,
and local chemical potential $\mu (x)$, organized into a ``big'' vector
$\beta_M = (\beta_\mu, \hat \mu)$, with $\beta^\mu = \beta (x) u^\mu (x)$ and $\hmu = \beta (x) \mu (x)$.
The noise variables can similarly be organized into a ``big'' vector $X_{a M} = (X_{a \mu}, \vp_a)$ with $X_a^\mu$ a spacetime vector and $\vp_a$ a scalar.
In the classical limit, the dynamical KMS transformation~\eqref{dkms} can be written as (with $\Th = \sP \sT$)
\be \label{oun}
\tilde \b_M (- x) = \b_M (x), \qquad  \p_\mu \tilde X_{a M} (-x) = \p_\mu \tilde X_{a M} (x)
+i \p_\mu \beta_M (x) \ .
\ee

Near equilibrium, the infinitesimal deviations from equilibrium values of these variables can be written as
\be
\chi_r^M=(\tau,u^i,\delta\mu), \quad \chi_a^M = (X_a^0, X_a^i, \vp_a), \quad M = 0,1,\cdots d-1, d , \quad
i=1, \cdots, d-1
\ee
where $\tau$ and $\delta\mu$ are defined through $\beta(x)^{-1}=T_0(1-\tau(x))$, $\mu(x)=\mu_0+\delta\mu(x)$,
with $T_0, \mu_0$ the equilibrium temperature and chemical potential. $u^i$ are the spatial components of
the (infinitesimal) velocity field (with $u^0 =1$). It is convenient to write the kernels $K_{MN}$ and $G_{MN}$ as
defined in~\eqref{lamn} in momentum space. One finds that $\xi_{r,a}^M$ factorize into two decoupled sectors:
the scalar sector $\xi_{r,a}^A=(\xi_{r,a}^0,\xi_{r,a}^z,\xi_{r,a}^d)$ and vector sector $\xi_{r,a}^\al$, with $\al=i \neq z$, where we have taken the spatial momentum $\vk$ to be along the $z$-direction.

After imposing the dynamical KMS symmetry, we find that, to lowest order in derivatives, for the vector sector
\be \label{nm}
K_{\al \b} = (i\omega T_0\p_T p_0-\eta k^2) \de_{\al \beta}, \qquad G_{\al \beta}= 2 T_0\eta k^2 \de_{\al \beta}
\ee
where $k = |\vk|$, $\eta$ is the shear viscosity, and $p_0$ is the equilibrium pressure density. For the scalar sector we have
\be
\begin{gathered}
K_{AB}=\begin{pmatrix} i\omega T_0^2\p_T^2 p_0 -\mu_0^2\sigma k^2& -ik T_0\p_T p_0 & i\omega T_0 \p_\mu \p_T p_0 +\mu_0\sigma k^2\\
-i k T_0 \p_T p_0 & i\omega T_0 \p_T p_0-\left(\zeta+2\frac{d-2}{d-1}\eta\right) k^2  & -ik \p_\mu p_0 \\
i\omega T_0 \p_T \p_\mu p_0 +\mu_0 \sigma k^2& -ik \p_\mu p_0 & i\omega  \p_\mu^2 p_0- \sigma k^2
\end{pmatrix}\\
G_{AB}=\begin{pmatrix} 2 T_0\mu_0 ^2\sigma k^2 & 0 & -2 T_0 \mu_0\sigma k^2\\
0& 2 T_0\left(\zeta+2\frac{d-2}{d-1}\eta\right)k^2 &0\\
-2 T_0\mu_0 \sigma k^2&0&2 T_0\sigma k^2
\end{pmatrix}\
\end{gathered}
\label{nm1}
\ee
where $\ze$ and $\sig$ are respectively bulk viscosity and conductivity. Equation~\eqref{pos} requires that
\be \label{jhn}
\sig \geq 0, \qquad \eta \geq 0, \qquad \ze \geq 0 \ .
\ee
It can be readily checked that for a thermodynamically stable system equations~\eqref{jhn} ensures all the poles of $K_{MN}^{-1}$ are indeed in the lower half $\om$-plane~\cite{CGL}.

\subsection{A simple illustration}

As a simple illustration, here we show that for the Brownian motion example of Sec.~\ref{sec:br}, two-point function of
$x_a$ vanishes at one-loop. We choose the potential term in (\ref{Leff-br}) to be $F(x)=f_1 x+f_2 x^2$. The Feynman rules are summarized in figure \ref{fig:feyn1b}.
\begin{figure}[!ht]
\centering
\includegraphics[width=11cm]{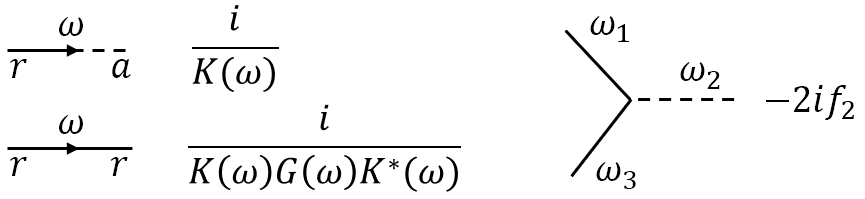}
\caption{Feynman rules for~\eqref{Leff-br}. Dashed lines correspond to $ x_a$, solid lines correspond to $ x$.}
\label{fig:feyn1b}
\end{figure}

\begin{figure}[!ht]
\centering
\includegraphics[scale=.8]{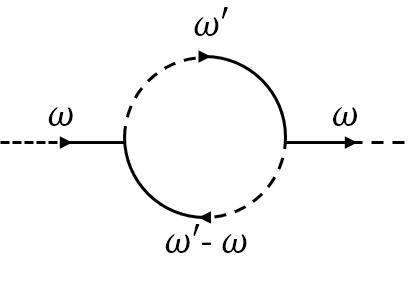}
\caption{One-loop contribution to the two-point function of $\protect x_a$.}\label{fig:1loop}
\end{figure}

The only one-loop contribution to $\vev{x_a (t) x_a (0)}$ comes from figure \ref{fig:1loop}, and gives
\be\label{loop1}
\langle x_a(-\omega)x_a(\omega)\rangle =
\frac i{K(\omega)}\frac i{K(-\omega)}\left[-2 f_2^2\int\frac{d\omega'}{2\pi}\frac i {K(\omega')}\frac i{K(\omega'-\omega)}\right]\ .\ee
From~(\ref{brpr}), the integrand in the square brackets of (\ref{loop1}) has four poles, all of which lie in the lower half-plane. Performing the integral by closing the contour in the upper half-plane, the integral vanishes. We then see that condition (\ref{Wc4}) is satisfied thanks to the analytic structure of the tree-level propagators.\footnote{For more examples of how condition (\ref{Wc4}) is preserved by loop corrections, see e.g. \cite{Kamenev,Arnold,Gonzalez}.}

\section{Proof of largest time equation } \label{sec:LTE}

In this section we shall prove that, for effective actions discussed in Sec.~\ref{sec:pert}, the LTE~(\ref{lte1}) is satisfied to all orders in perturbation theory.


\begin{theorem}
Given any Lagrangian of the form~\eqref{lamn} with a retarded $K_{ij}$, the largest time equation~\eqref{lte1} is satisfied for any operators of the form~\eqref{opn}.
\end{theorem}

\begin{proof}[Proof:]

We will use the regulator~\eqref{mm} to cut off frequency integrals, which ensures~\eqref{rap} for an $a-r$ propagator, and the parallel statement for a $r-a$ propagator. For a theory with a finite number of types of interaction vertices, we can choose $p$ of~\eqref{mm} to be greater than $2 m_{\rm max}$, where $m_{\rm max}$ denotes the largest number of time derivatives that a vertex
can have.\footnote{When introducing the regulator (\ref{mm}), one may need to modify the interacting part of the Lagrangian $\mathcal L_I$ as well, when there is a nonlinear symmetry. An example is the dynamical KMS symmetry~\eqref{oun}.
In such a case, one should expand the regulator in $\sL_I$ in derivatives and keep the total number of time derivatives to be
$m_{\rm max}$. This guarantees that to such derivative order nonlinear symmetries are preserved.} For such a $p$,~\eqref{rap} is then satisfied for all loop diagrams.

In~\eqref{lte1} the operator with the largest time is an $a$-type operator which contains at least one factor of $\chi_a$.
Start with an external $\chi_a$, which can only connect through an $a-r$ propagator, the other end of the propagator
must be a $r$-leg of a vertex. Since all vertices must contain at least one $a$-type leg, let us continue along an $a$-leg
which in turn must connect with a $r$-leg through an $a-r$ propagator.
Repeating this procedure we will end up with the following possibilities:

\ben

\item We run into an internal vertex that we already encountered. In this case we have an internal loop whose every propagator is of $a-r$ type, see Fig.~\ref{fig:plot}(a). Consider the internal $a-r$ loop, then from~\eqref{rap} we will have a product of the form
\be\label{ww1}
\th (w_2 - w_1) \th (w_3 - w_2) \cdots \th (w_1 -w_m) W
\ee
with $W$ denoting products of functions $F$ as in~\eqref{rap} associated with each propagator.
The whole diagram is thus zero unless $w_1 = w_2 = \cdots = w_m$, but in this latter case from~\eqref{rap1} $W$ is zero
due to that the $F$'s vanish for coincidental times.

\item We come back to the same external vertex we started with through a $r$-leg, see Fig.~\ref{fig:plot}(b). In this case we have a product of the form
\be\label{ww2}
\th (w_1 - t_{n+1}) \th (w_2 - w_1)   \cdots \th (t_{n+1}- w_m) W  \ .
\ee
Following the same reasoning as that for the previous case we conclude that the contribution is identically zero.

\item We end up at another external vertex, see Fig.~\ref{fig:plot}(c). In this case we have
\be\label{ww3}
\th (w_1 - t_{n+1}) \th (w_2 - w_1)  \cdots \th (t_k - w_m) W , \quad {\rm with} \quad k \in \{1, 2, \cdots, n\} \ .
\ee
which requires $t_k \geq t_{n+1}$ which is contradictory with $t_{n+1}$ being the largest time. Thus it is also identically zero.

\een
Thus we find in all cases the correlation function is identically zero.
\end{proof}

\begin{figure}[!ht]
\begin{centering}
\includegraphics[width=5cm]{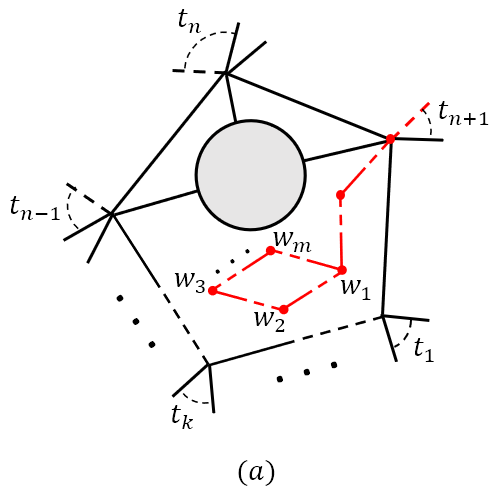} \includegraphics[width=5cm]{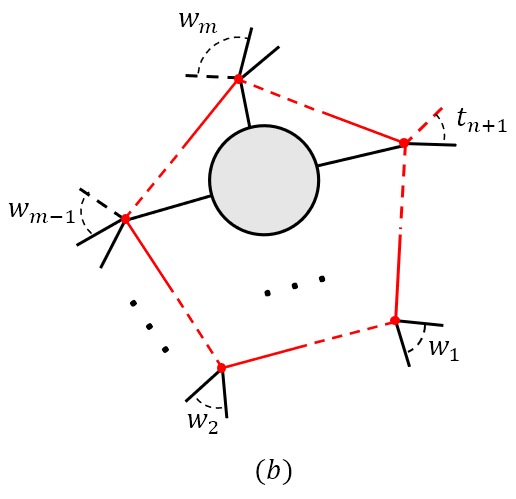}
\includegraphics[width=5cm]{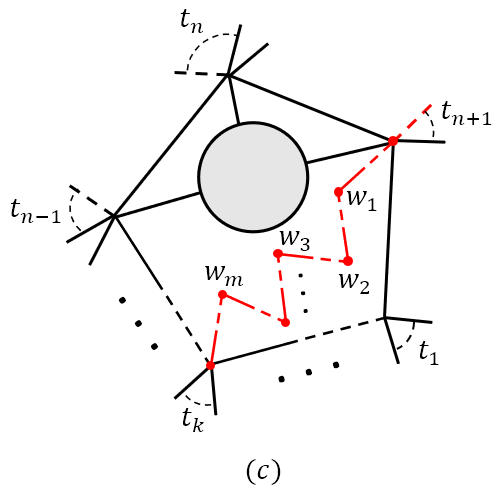}
\caption{Various possibilities: (a) corresponds to the existence of an internal $a$-$r$ loop, which is described by~\eqref{ww1}.
In (b), there is an $a$-$r$ loop which starts and ends at the operator with the largest time, which is described by~\eqref{ww2}.
In (c) there is an $a$-$r$ line which connects the operator with the largest time with another operator, described by~\eqref{ww3}. }\label{fig:plot}
\end{centering}
\end{figure}

We note that the essence of the above proof is closely related to the all-loop proof of the normalization condition (\ref{Wc4})
in a class of open quantum field theory in~\cite{loga}, as well as the examples discussed in \cite{Kamenev}. It is also close in spirit to the discussions of the path integral of the Langevin equation \cite{Arnold,Gonzalez}.

\section{Contributions from BRST ghosts} \label{sec:ghost}

In the last section we proved that given the advanced nature of the $a-r$ propagators, the LTE~\eqref{lte1} is satisfied
to all orders in perturbation theory. The proof also warrants the normalization condition~\eqref{Wc3} and~\eqref{Wc4}, which is a subcase of~\eqref{lte1}. At the end of Sec.~\ref{sec:eft} we briefly mentioned another method  to guarantee the normalization condition by adding BRST ghosts and extending the original bosonic action to a BRST invariant action.
This method is independent on the pole structure.

In this section we show that after taking into account the retarded structure of $r-a$ propagators all ghost contributions identically vanish in the regularization scheme of~\eqref{mm}. See also~\cite{Arnold,Gonzalez} for similar discussions.



 As shown in~\cite{GaoL}, to guarantee KMS invariance,  the BRST extended quadratic action~\eqref{lamn}
 must take the form
\be
\mL_0=\chi_{a\rmi} K_{\rmi\rmj} \chi_{r\rmj} +\f{i}{2}\chi_{a\rmi} G_{\rmi\rmj} \chi_{a\rmj} - c_{a\rmi} K_{\rmi\rmj} c_{r\rmj} \ .
\ee
The above equation means that ghosts only have $r$-$a$ and $a$-$r$ types propagators,
which are the same as the corresponding bosonic fields. Furthermore, BRST symmetry~\eqref{omm} implies that
in the BRST-extended Lagrangian terms involving ghosts \textit{must contain at least one factor of $c_a$}. Now consider integrating out the ghost variables in the path integral
\be
e^{W [\phi_1, \phi_2]}  =  \int D \chi_{r,a} D c_{r,a} \, e^{i I_{\rm eff} [\chi_r, \chi_a, c_r, c_a; \phi_r, \phi_a]}   \
\label{left1}
\ee
We can use the same arguments as in last section: (i) all diagrams involving ghosts must contain at least one $a$-$r$ loop, i.e. there exists a loop in which every propagator is an $a$-$r$ type ghost propagator.  (ii) Such an $a$-$r$ loop vanishes
identically as in the discussion around equation~\eqref{ww1}. We thus conclude that to all orders in perturbation theory ghosts do not make any contribution.

\section{Other unitarity constraints and KMS conditions}  \label{sec:uni}

\subsection{Other unitarity constraints}

We now briefly comment on the fulfillment of two other unitarity constraints~\eqref{Wc1} and~\eqref{Wc2}
by effective theory path integral~\eqref{left}.

For~\eqref{Wc1}, it can be readily to show that given~\eqref{fer1} it holds to all loops~\cite{CGL}. Here we present a slightly different argument. Equation~\eqref{fer1} implies that we can rearrange the action as
\begin{equation}
\mL_{\rm eff}=i\sum_{m,n} g_{nm} \Phi^{n}(i\Phi_{a})^{m}
\label{act13}
\end{equation}
where $\Phi = (\phi_r, \chi_r)$ and $\Phi_{a}=(\phi_a, \chi_a)$ collectively denote both source and dynamical fields
and $g_{nm}$ are some {\it real} differential operators. Defining a new field $\tilde{\Phi}_{a}=i\Phi_{a}$ it is clear that
integrating out dynamical fields $\chi_{r,a}$, the resulting $W$ will again have the structure~\eqref{act13} as all propagator and vertices are real (in coordinate space).

Now let us consider~\eqref{Wc2}. We will show that given~\eqref{pos} it holds perturbatively in loop expansion and infinitesimal sources.
For this purpose let us write $\sL_{\rm eff}$ as
\be \label{onm}
\sL_{\rm eff} = \sL_{\rm saddle} [\phi_r, \phi_a] + \sL_{\rm fluct} [\chi_r, \chi_a; \phi_r, \phi_a]\ .
\ee
In writing down~\eqref{onm} we have assumed that the Lagrangian density has been expanded around  whatever saddle-point solution one is interested in, i.e. $\sL_{\rm saddle}$ is the expression of $\sL_{\rm eff}$ evaluated at the solution, and  $\chi_{r,a}$ should be considered as deviations around the solution.

 We thus have
 \be
 W = W_{\rm tree} + W_{\rm loop} [\phi_r, \phi_a] , \qquad W_{\rm tree} = i \int d^d x \, \sL_{\rm saddle} [\phi_r, \phi_a]
 \ee
 where $\sL_{\rm saddle}$ gives the tree-level part of $W$ while $W_{\rm loop}$ is obtained by integrating out fluctuations of dynamical fields
 $\chi_{r,a}$
 \be
e^{W_{\rm loop}} = \int D \chi_r D \chi_a \, e^{i \int d^d x \, \sL_{\rm fluct} } \ .
\ee
From the usual loop counting argument
\be
\sL_{\rm eff} \sim {1 \ov \hbar_{\rm eff}} \quad \implies \quad W_{\rm tree}  \sim {1 \ov \hbar_{\rm eff}}, \qquad
W_{\rm loop} \sim O(1) + O(\hbar_{\rm eff}) + \cdots
\ee
where $\hbar_{\rm eff}$ is the effective loop counting parameter and is taken to be small in perturbation theory. Thus in perturbation theory $W_{\rm tree}$ always dominates over $W_{\rm loop}$.

Now $W_{\rm tree}$ can be further expanded in powers of $\phi_r, \phi_a$ which we take to be infinitesimal, and perturbatively the non-positivity of ${\rm Re} W_{\rm tree}$ is dictated by that of the quadratic terms $W^{(2)}_{\rm tree}$ as follows. Writing
\be \begin{gathered}
{\rm Re} \, W^{(2)}_{\rm tree} =\int d^d x_1 d^d x_2\, A_2(x_1,x_2)\phi_a(x_1)\phi_a(x_2)\\ {\rm Re} \, W^{(3)}_{\rm tree} = \int d^d x_2 d^d x_2 d^d x_3 \, A_3(x_1,x_2,x_3)\phi_a(x_1)\phi_a(x_2)\phi_r(x_3)\ ,\end{gathered}
\ee
then to quartic order we can rewrite  ${\rm Re} \, W_{\rm tree}$ by  ``completing the square'' as
\be \begin{gathered}{\rm Re} \, W_{\rm tree}
= \int d^d x_1 d^d x_2\,A_2(x_1,x_2)\left(\phi_a(x_1)+\int d^d x_3 d^d x_4\, H(x_1,x_3,x_4)\phi_a(x_3)\phi_r(x_4)\right)\\
\times \left(\phi_a(x_2)+\int d^d x_5 d^d x_6 \, H(x_2,x_5,x_6)\phi_a(x_5)\phi_r(x_6)\right)+ R_4\ ,
\end{gathered}\ee
where $R_4$ is quartic in $\phi$'s and $H$ is obtained by convoluting the inverse of $A_2$ with $A_3$. Higher powers in $\phi$'s can be treated similarly.

Then we are left to show ${\rm Re} \, W^{(2)}_{\rm tree} $ is non-positive. For this purpose, let us note that the quadratic action has the general form
\be S_0=\chi_a^\dag K_1 \chi_r+\frac i2 \chi_a^\dag G_1 \chi_a+\phi_a^\dag K_2 \chi_r+i\phi_a^\dag G_2 \chi_a+\phi_r^\dag K_3 \chi_a + \phi_a^\dag K_4 \phi_r+\frac i2 \phi_a^\dag G_3 \phi_a\ ,\ee
where e.g. $\chi_a^\dag K_1 \chi_r=\int d^{d-1} kd\omega\, \chi_{ai}^*(\omega,\vec k)K_{1ij}(\omega,\vec k) \chi_{rj}(\omega,\vec k)$. The tree-level contribution to $W^{(2)}_{\rm tree}$ comes from evaluating $S_0$ on the saddle point
\be \chi_a=-K_1^{-\dag}K_2^\dag\phi_a,\quad \chi_r=-K_1^{-1}K_3^\dag\phi_r+i K_1^{-1}(G_1 K_1^{-\dag}K_2^\dag-G_2)\phi_a\ ,\ee
giving
\be {\rm Re} \, W^{(2)}_{\rm tree}=-\frac 12 \phi_a^\dag(K_2K_1^{-1}G_1 K_1^{-1 \dag}K_2^\dag+G_3-2G_2K_1^{-1 \dag}K_2^\dag)\phi_a\ . \ee
Now with $\text{Im}\, S_0\geq 0$ gives $ {\rm Re} \, W^{(2)}_{\rm tree} \leq 0$.

\subsection{KMS conditions}

The dynamical KMS symmetry~\eqref{lkms} can be considered as a mathematical definition
of systems in local equilibrium, i.e. a general initial state $\rho_0$ can be considered as describing a local equilibrium system
only if the corresponding EFT possesses the dynamical KMS symmetry.

A special case is when $\rho_0$ is a thermal density matrix, for which the corresponding EFT describes the dynamics of (nonlinear) disturbances around thermal equilibrium. In this case, correlation functions satisfy in addition the Kubo-Martin-Schwinger (KMS) relations, which when combined with a time reversal $\Th$ (as that defined in~\eqref{dkms}), can be succinctly expressed as a $Z_2$ symmetry of the generating functional~\eqref{gen0}~\cite{CGL}
\be \label{okms}
W[\phi_r,\phi_{a}] = W[\tilde{\phi}_r,\tilde{\phi}_{a}]
\ee
where (we take the classical limit for simplicity)
\be \label{skms}
\Th \tilde \phi_r (x) =  \phi_r (x), \qquad \Th \tilde \phi_a (x) = \phi_a (x) + i \b_0 \p_t \phi_r (x)
\ee
and $\b_0$ is again the inverse equilibrium temperature. In the presence of external sources, we require the effective action $I_{\rm eff} [\chi_r, \chi_a; \phi_r, \phi_a]$ to be invariant under both~\eqref{dkms} and~\eqref{skms}.  Note that invariance under~\eqref{skms} constrains the structure of contact terms of sources in the effective action. These contact terms are important, e.g. they contribute to susceptibilities and Kubo formulas.

Eq.~\eqref{okms} will be automatically satisfied by~\eqref{left} as
\bea
e^{W[\phi_r,\phi_{a}]} &\equiv & \int D\chi_r D\chi_{a}\; e^{iI_{\rm eff} [\chi_r,\chi_{a};\phi_r,\phi_{a}]}=\int D\chi_r D\chi_{a} \, e^{iI_{\rm eff} [\tilde{\chi}_r,\tilde{\chi}_{a};\tilde{\phi}_r,\tilde{\phi}_{a}]} \cr
&= & \int D\tilde{\chi}_r D\tilde{\chi}_{a} \, e^{iI_{\rm eff} [\tilde{\chi}_r,\tilde{\chi}_{a};\tilde{\phi}_r ,\tilde{\phi}_{a}]}=e^{W[\tilde{\phi}_r,\tilde{\phi}_{a}]} \label{kmsproof}
\eea
where the change of integration measure from $\chi_r$ and $\chi_a$ to $\tilde{\chi}_r$ and $\tilde{\chi}_a$ has determinant one.
The manipulations in~\eqref{kmsproof} hold provided that the regularization procedure one uses to make the path integrals finite is compatible with transformations~\eqref{dkms} and~\eqref{skms}. Indeed this was a main motivation to use the regularization procedure~\eqref{mm}. In contrast, a hard frequency cutoff will not be compatible, and one will need to include non-KMS invariant counter terms as can be readily checked in explicit examples.

\section{Discussion and conclusions} \label{sec:conc}

Let us summarize the main results of the paper. The largest time equation is a consequence of unitarity and implies causality. It is  a key constraint on path integrals along a CTP contour. Any effective field theory must respect the LTE. In this paper, we proved a theorem showing that if the propagators of dynamical fields of the effective action have the proper pole structure, LTE is obeyed to all loop orders. 
Using the same arguments we also showed that all ghost contributions are trivial. We should emphasize that dynamical KMS invariance was not directly used to
prove the LTE. It was used to guarantee that the $r$-$a$
propagators have the retarded structure. If the retarded  property arises from other requirements,
which we expect to be the case even for theories not in local thermal
equilibrium,  the LTE can still be proved.

We reached the conclusion by using a specific regularization procedure~\eqref{mm} which has the advantages of maintaining the retarded structure of $K_{ij}^{-1}$ and being compatible with dynamical KMS transformation.
On general grounds one expects that, if some other regularization scheme is used, one will reach the same conclusion up to possible  local counter terms to be added to the bare Lagrangian and/or local contact terms in correlation functions.
We have checked that this is indeed the case for sharp cutoff in frequency integrals. 

Finally we should mention that ghosts and supersymmetry are still useful if one prefers to use other type of regulators which break the retarded structure of the $r$-$a$ propagators or dynamical KMS symmetry.\footnote{It should be kept in mind that supersymmetry only imposes one particular constraint from the dynamical KMS symmetry. One still has to make sure the rest
is properly imposed~\cite{GaoL}.} They will help to ensure the normalization condition and part of the dynamical KMS symmetry to be
manifestly preserved.


\vspace{0.2in}   \centerline{\bf{Acknowledgements}} \vspace{0.2in}
We would like to thank Kristan Jensen for suggestions on the title and for conversations. We would also like to thank Derek Teaney for valuable clarifications on existing literature, and Misha Stephanov and Amos Yarom for conversations. This work is supported by the Office of High Energy Physics of U.S. Department of Energy under grant Contract Number  DE-SC0012567. P. G. was supported by a Leo Kadanoff Fellowship.

\end{document}